\documentclass[letterpaper]{article}
\usepackage{preprint}
\usepackage{setspace}
\usepackage[utf8]{inputenc}
\usepackage[T1]{fontenc}
\usepackage{microtype}      
\usepackage{booktabs}       

\usepackage{amsmath,amsfonts,amssymb,amsthm}
\usepackage{mathtools}

\usepackage{enumerate}
\usepackage{cite}
\usepackage{url}
\usepackage{xcolor}

\usepackage{graphicx}
\DeclareGraphicsExtensions{.eps}

\usepackage{calc}

\usepackage{tikz}
\usetikzlibrary{shapes, decorations.markings, decorations.pathreplacing, calligraphy}

\newtheoremstyle{mythm}
  {5pt}
  {5pt}
  {}
  {}
  {\bfseries}
  {.}
  {3pt}
  {}

\theoremstyle{mythm}
\newtheorem{theorem}{Theorem}

\newtheorem{corollary}{Corollary}[theorem]
\newtheorem{lemma}[theorem]{Lemma}

\newtheorem{assumption}{Assumption}

\theoremstyle{remark}

\def \ba{\begin{array}}
\def \ea{\end{array}}
\def \be{\begin{equation}}
\def \ee{\end{equation}}

\def \bb{\mathbb}

\def \mc{\mathcal}

\def \diag{\mathrm{diag}}
\def \rank{\mathrm{rank}}

\def \sym{\mathrm{sym}}
\def \T{\intercal} 

\def \S{\mathtt{S}}
\def \I{\mathtt{I}}

\def \R{\mathtt{R}}

\usepackage{geometry}
\newgeometry{top=1in,bottom=0.75in,right=0.75in,left=0.75in}

\makeatletter
\def\endthebibliography{%
  \def\@noitemerr{\@latex@warning{Empty `thebibliography' environment}}%
  \endlist
}
\makeatother

\title{Parameterization-Free Observer Design for Nonlinear Systems: Application to the State Estimation of Networked SIR Epidemics
\thanks{This work is supported by the European Union's Horizon Research and Innovation Programme under Marie Sk\l{}odowska-Curie grant agreement No. 101062523. It has also received funding from the Swedish Research Council and the Knut and Alice Wallenberg Foundation, Sweden.}}

\author{
Muhammad Umar B. Niazi
\thanks{Laboratory for Information \& Decision Systems, Massachusetts Institute of Technology, Cambridge, MA 02139, USA. Email: \texttt{niazi@mit.edu}}
\And
Karl H. Johansson
\thanks{Division of Decision and Control Systems, Digital Futures, KTH Royal Institute of Technology, SE-100 44 Stockholm, Sweden. Email: \texttt{kallej@kth.se}}
}

\begin{document}
\maketitle

\vspace{1cm}
\begin{abstract}
Traditional observer design methods rely on certain properties of the system's nonlinearity, such as Lipschitz continuity, one-sided Lipschitzness, a bounded Jacobian, or quadratic boundedness. These properties are described by parameterized inequalities. However, enforcing these inequalities globally can lead to very large parameters, resulting in overly conservative observer design criteria. These criteria become infeasible for highly nonlinear applications, such as networked epidemic processes. In this paper, we present an observer design approach for estimating the state of nonlinear systems, without requiring any parameterization of the system's nonlinearities. The proposed observer design depends only on systems' matrices and applies to systems with any nonlinearity. We establish different design criteria for ensuring both asymptotic and exponential convergence of the estimation error to zero. To demonstrate the efficacy of our approach, we employ it for estimating the state of a networked SIR epidemic model. We show that, even in the presence of measurement noise, the observer can accurately estimate the epidemic state of each node in the network. To the best of our knowledge, the proposed observer is the first that is capable of estimating the state of networked SIR models.
\end{abstract}

\newpage
\section{Introduction}

Nonlinear systems are ubiquitous in engineering, physics, and biology. Accurately estimating their states is critical for many applications, such as output feedback control, fault diagnosis, and prediction. Observers, typically used for these purposes, have been extensively studied over the past few decades in the control systems community.

Traditional methods for observer design rely on the assumption that the nonlinearity of the system satisfies certain properties, which are defined through parameterized inequalities. For example, Luenberger-like observers, proposed by Thau in \cite{thau1973}, have been designed using such parameterization-based techniques. Various observer design criteria have been proposed in the literature, such as Lipschitz continuity \cite{rajamani1998, zemouche2013, alessandri2020, niazi2022cdc}, one-sided Lipschitz and quadratic inner-boundedness \cite{abbaszadeh2010,zhao2010,benallouch2012}, bounded Jacobian \cite{phanomchoeng2011, tahir2021}, quadratic boundedness \cite{siljak2000}, passivity \cite{arcak2001}, and dissipativity \cite{moreno2004}. However, to ensure the global convergence of the observer, these parameterizations have to be enforced globally, which results in very large parameters \cite{nugroho2021}. Because of this, parameterization-based techniques lead to unnecessarily conservative observer design criteria, limiting their applicability to highly nonlinear systems.

To overcome this challenge, we propose a novel observer design approach that does not require any assumptions about the system's nonlinearity, making it parameterization-free. The proposed approach essentially treats the nonlinearity as an unknown disturbance in the estimation error dynamics. Therefore, the goal of the observer is to not only track the state but also the nonlinearity of the system. We establish observer design criteria that guarantee both asymptotic and exponential convergence of the estimation error to zero.

We demonstrate the effectiveness of our approach through the state estimation of a networked SIR epidemic model. State estimation of epidemic models is critical because the infectious population is usually difficult to measure directly and needs to be estimated. Moreover, parameterization-based design becomes infeasible for epidemic models because their nonlinearity, i.e., the mass action law, is quadratic in nature, rendering the Lipschitz constant and other parameterizations to be very large. Through simulations, we show that the observer can effectively estimate the epidemic state in every node even in the presence of measurement noise.

The main contributions of this paper include a parameterization-free observer design approach and its application to networked SIR epidemic processes. Notably, this work is the first to propose an observer capable of estimating the state of a networked SIR epidemic model, to the best of our knowledge. Our previous work on the parameterization-based observer design \cite{niazi2022cdc} can estimate the state of population models of epidemic processes. Additionally, \cite{mei2022} demonstrated that their parameterization-based observer can estimate the state of a networked SIS epidemic model. However, neither of these approaches can estimate the state of networked SIR models. Therefore, our proposed parameterization-free approach not only offers a promising alternative to traditional observer design methods but also has the potential to be applied to a wide range of nonlinear systems.

The paper is organized as follows. Section~\ref{sec:prob} outlines the problem statement, while Section~\ref{sec:observer_form} introduces the observer form utilized in this paper. In Section~\ref{sec:par-based}, we briefly review the parameterization-based designs of this observer and highlight their limitations. Next, in Section~\ref{sec:par-free}, we present our parameterization-free observer design. The effectiveness of our approach is demonstrated in Section~\ref{sec:simulation}, where we apply it to a networked SIR epidemic model. Finally, we conclude in Section~\ref{sec:conclusion} with our closing remarks.

\section{Problem Statement}
\label{sec:prob}

We consider nonlinear systems that can be described as
\begin{subequations}
    \label{eq:sys}
    \begin{align}
        \label{eq:sys_state}
        \Dot{x} &= Ax + Gf(Hx) \\
        \label{eq:sys_output}
        y &= Cx
    \end{align}
\end{subequations}
where $x(t)\in\mc{X}\subseteq\bb{R}^{n_x}$ is the state and $y(t)\in\bb{R}^{n_y}$ is the measured output. The nonlinear function $f:H\mc{X}\rightarrow\bb{R}^{n_f}$ and matrices $A\in\bb{R}^{n_x\times n_x}$, $G\in\bb{R}^{n_x\times n_f}$, $H\in\bb{R}^{n_h\times n_x}$, and $C\in\bb{R}^{n_y\times n_x}$ are known. Notice that the control input is omitted in \eqref{eq:sys} only for brevity and without loss of generality.

We assume that the system \eqref{eq:sys} is asymptotically detectable. That is, for any pair of solution trajectories $t\mapsto x(t;x_a)$ and $t\mapsto x(t;x_b)$ initialized from $x_a,x_b\in\mc{X}$ and defined on $t\in[0,\infty)$, it holds that $y(t;x_a)=y(t;x_b)$ implies $\lim_{t\rightarrow\infty}\|x(t;x_a)-x(t;x_b)\|=0$. This assumption is a necessary condition for the existence of an asymptotic state observer \cite{bernard2022}. In addition, we also assume that $(A,C)$ is a detectable pair, i.e.,
\[
\rank\left[\ba{c} sI-A \\ C \ea\right] = n, \quad \forall s\in\bb{C}_{\geq 0}.
\]
It is important to remark that, if the above condition is not satisfied for $(A,C)$, one can always add and subtract $\Delta$ so that $(A+\Delta,C)$ satisfies the above condition. In this case, we can either choose
\[
\bar{G}\bar{f}(\bar{H}x)=\underbrace{\left[\ba{cc} G & -\Delta \ea\right]}_{\bar{G}}\underbrace{\left[\ba{c} f(Hx) \\ x \ea\right]}_{\bar{f}}
\] 
where $\bar{H}=I_{n_x}$, or
\[
\bar{G}\bar{f}(\bar{H}x)=\underbrace{Gf(Hx)-\Delta x}_{\bar{f}}
\]
where $\bar{G}=I_{n_x}$ and $\bar{H}=I_{n_x}$.

Subject to the above assumptions, our goal is to design an observer 
\[
\Dot{z}=\phi(z,y), \quad \hat{x}=\theta(z,y)
\]
such that the estimation error $\xi(t)=\hat{x}(t)-x(t)$ satisfies
\be
\label{eq:KL_ineq}
\|\xi(t)\| \leq \beta(\|\xi(0)\|,t)
\ee
where $\beta:\bb{R}_{\geq 0}\times\bb{R} \rightarrow \bb{R}_{\geq 0}$ is some class-$\mc{KL}$ function. In other words, we want $\lim_{t\rightarrow+\infty}\|\xi(t)\|=0$ irrespective of the initial error $\xi(0)\in\bb{R}^{n_x}$.

\section{Proposed Form of the Observer}
\label{sec:observer_form}

We consider the observer form proposed in \cite{niazi2022cdc}
\begin{subequations}
    \label{eq:obs}
    \begin{align}
        \Dot{z} &= Mz + (ML+J)y + Nf(\eta) \label{eq:obs_state} \\
        \eta &= H\hat{x} + K(y-C\hat{x}) \label{eq:obs_eta} \\
        \hat{x} &= z + Ly \label{eq:obs_estimate}
    \end{align}
\end{subequations}
where $z(t)\in\bb{R}^{n_x}$ is the observer's state and $\hat{x}(t)\in\bb{R}^{n_x}$ is its output. The matrices $M\in\bb{R}^{n_x\times n_x}$ and $N\in\bb{R}^{n_x\times n_f}$ are chosen as
\begin{equation}
    \label{eq:matrices_M,N}
    M = A-LCA-JC, \quad N = I-LC
\end{equation}
whereas $J,L\in\bb{R}^{n_x\times n_y}$ and $K\in\bb{R}^{n_h\times n_y}$ are gain matrices that need to be designed.

Define the estimation error $\xi(t) \doteq \hat{x}(t)-x(t)$, then the error equation is given by
\begin{equation}
    \label{eq:error}
    \Dot{\xi} = M \xi + NG\Tilde{f}(\eta,Hx)
\end{equation}
where
\begin{equation}
    \label{eq:Tilde_f}
    \Tilde{f}(\eta,Hx) \doteq f(\eta) - f(Hx).
\end{equation}
Thus, our design objective is to find the gain matrices $J,K,L$ such that the estimation error satisfies \eqref{eq:KL_ineq}. 

Although the observer \eqref{eq:obs} has a different form than traditional Luenberger-like observers \cite{boutat2021}, it is inspired by the observer form developed by Luenberger \cite{luenberger1971}. This form of the observer offers additional degrees of freedom for designing $L$ in a way that minimizes the impact of $\Tilde{f}$ in the error dynamics \eqref{eq:error}. Also, the innovation term $K(y(t)-C\hat{x}(t))$ in the function $f(\eta(t))$ in \eqref{eq:obs_state}-\eqref{eq:obs_eta} allows for designing $K$ such that the observer can effectively track the nonlinear signal $f(Hx(t))$ in \eqref{eq:sys_state}. Furthermore, an additional gain matrix $J$ can be utilized to achieve stability of the matrix $M=A-LCA-JC$.

The following lemma will be used in our analysis. It is quite standard in the Lyapunov stability theory, and the reader is referred to \cite[Chapter 5]{sastry1999} for more details and the proof.

\begin{lemma} \label{lem:lyap}
    If there is a function $V:\bb{R}^{n_x}\rightarrow\bb{R}_{\geq 0}$ that, for every $\xi(t)\in\bb{R}^{n_x}$, satisfies
    \begin{enumerate}[(i)]
        \item $\gamma_1\|\xi(t)\|^2 \leq V(\xi(t)) \leq \gamma_2 \|\xi\|^2$, for some $\gamma_2>\gamma_1>0$
        \item $\Dot{V}(\xi(t)) < 0$
    \end{enumerate}
    then \eqref{eq:error} is uniformly, globally asymptotically stable. In other words, irrespective of the signal $\Tilde{f}(\eta(t),x(t))$ in the error equation \eqref{eq:error}, there exists a class-$\mc{KL}$ function $\beta$ such that \eqref{eq:KL_ineq} holds.
    
    Moreover, if, instead of (ii), $\Dot{V}(\xi(t)) < -\alpha V$, for some $\alpha>0$, then \eqref{eq:error} is uniformly, globally exponentially stable. In other words, irrespective of the signal $\Tilde{f}(\eta(t),x(t))$ (defined in \eqref{eq:Tilde_f}), there exist constants $c\doteq c(\|\xi(0)\|)>0$ and $\lambda>0$ such that \eqref{eq:KL_ineq} holds with $\beta(\|\xi(0)\|,t)=c(\|\xi(0)\|)e^{-\lambda t}$.
    \hfill $\Box$
\end{lemma}

The next section introduces observer design criteria that rely on the Lipschitz continuity and quadratic boundedness of the nonlinear function $f$. This approach is known as parameterization-based design. Then, in the subsequent section, we will propose a parameterization-free design, which does not rely on any specific parameterization of $f$.

\section{Review of Parameterization-based Observer Design and its Limitations}
\label{sec:par-based}

In this section, we briefly review existing results related to the parameterization-based design of observers, with a particular focus on two commonly used parameterizations: Lipschitz continuity and quadratic boundedness. However, it is important to note that enforcing these parameterizations globally can have limitations, which will also be discussed.

\begin{assumption}
    \label{assum:lipschitz}
    The function $f$ is Lipschitz continuous on $\mc{X}$. That is, there exists $\ell\in\bb{R}_{\geq 0}$ such that, for every $x,\hat{x}\in\mc{X}$,
    \begin{equation}
        \label{eq:lipschitz}
        \|f(H\hat{x})-f(Hx)\| \leq \ell \|H(\hat{x}-x)\|.
    \end{equation}
\end{assumption}

It is well-known that if $f$ is continuously differentiable, then $\ell$ can be computed by solving the following nonlinear optimization problem
\begin{equation}
    \label{eq:lip_compute}
    \ell=\sup_{x\in\mc{X}} \left\| \frac{\partial f}{\partial x} (Hx) \right\|.
\end{equation}

\begin{theorem}[Niazi \& Johansson \cite{niazi2022cdc}]
    \label{thm:lip}
    Let Assumption~\ref{assum:lipschitz} hold. If there exist a positive definite matrix $P=P^\T\in\bb{R}^{n_x}$ and gain matrices $J,L\in\bb{R}^{n_x\times n_y}$ and $K\in\bb{R}^{n_h\times n_y}$ such that
    \begin{equation}
        \label{eq:ari_lip}
        M^\T P + PM + PNGG^\T N^\T P + \ell^2 (H-KC)^\T (H-KC) <0
    \end{equation}
    where $M,N$ are given in \eqref{eq:matrices_M,N}, then the estimation error $\xi(t)$ globally asymptotically converges to zero, i.e., $\lim_{t\rightarrow\infty}\|\xi(t)\|=0$.
\end{theorem}

The algebraic Riccati inequality (ARI) \eqref{eq:ari_lip} guarantees global asymptotic stability of the error dynamics \eqref{eq:error}. To guarantee global exponential stability under Lipschitz property, one can adapt \eqref{eq:ari_lip} to
\begin{equation}
    \label{eq:ari_lip_exp}
    (M+\alpha I_{n_x})^\T P + P(M+\alpha I_{n_x}) + PNGG^\T N^\T P
    + \ell^2 (H-KC)^\T (H-KC) <0.
\end{equation}
for some $\alpha>0$. Moreover, note that both \eqref{eq:ari_lip} and, given $\alpha>0$, \eqref{eq:ari_lip_exp} can be equivalently represented as linear matrix inequalities using Schur complement lemma.

\begin{assumption}
    \label{assum:qb}
    The function $f$ is quadratically bounded on $\mc{X}$. That is, there exists a positive definite matrix $Q\in\bb{R}^{n_x\times n_x}$ such that, for every $x,\hat{x}\in\mc{X}$,
    \begin{equation}
        \label{eq:qb}
        \|f(H\hat{x})-f(Hx)\|^2 \leq (\hat{x}-x)H^\T Q H(\hat{x}-x).
    \end{equation}
\end{assumption}

For general nonlinear functions $f$, finding a matrix $Q$ such that \eqref{eq:qb} holds is a difficult problem. However, if $f$ is continuously differentiable, then we can employ a method proposed by \cite{nugroho2021} to compute a diagonal $Q=\diag(q_1,\dots,q_{n_x})$ by solving the following nonlinear optimization problem:
\begin{equation}
    \label{eq:qb_compute}
    q_i=\sup_{x\in\mc{X}} n_x \sum_{j=1}^{n_x} \left(\frac{\partial f_j}{\partial x_i} (Hx) \right)^2.
\end{equation}

\begin{theorem}
    Let Assumption~\ref{assum:qb} hold. If there exist a positive definite matrix $P=P^\T\in\bb{R}^{n_x}$ and gain matrices $J,L\in\bb{R}^{n_x\times n_y}$ and $K\in\bb{R}^{n_h\times n_y}$ such that
    \begin{equation}
        \label{eq:ari_qb}
        M^\T P + PM + PNGG^\T N^\T P + (H-KC)^\T Q (H-KC) <0
    \end{equation}
    where $M,N$ are given in \eqref{eq:matrices_M,N}, then the estimation error $\xi(t)$ globally asymptotically converges to zero, i.e., $\lim_{t\rightarrow\infty}\|\xi(t)\|=0$.
\end{theorem}
\begin{proof}
    The proof is similar to that of Theorem~\ref{thm:lip} in \cite{niazi2022cdc}.
\end{proof}

The Lipschitz continuity \eqref{eq:lipschitz} and quadratic boundedness \eqref{eq:qb} parameterizations bound the incremental change in the nonlinear function $f$ from above in terms of the incremental change in its arguments. By enforcing these inequalities over the whole state space $\mc{X}$, it turns out that the Lipschitz constant $\ell$ and/or quadratic boundedness parameter $Q$ are very large. This restricts the possibility to choose the gain matrices such that they satisfy ARIs \eqref{eq:ari_lip} and/or \eqref{eq:ari_qb}. For instance, in the case of \eqref{eq:ari_qb}, the equivalent LMI feasibility problem is
\begin{align*}
    \left[\ba{cc} \sym(PA-RCA-SC)+T & (P-RC)G \\ * & -I_{n_x} \ea\right] &< 0 \\
    \left[\ba{cc} -T & (H-KC)^\T \\ * & -Q^{-1} \ea\right] &\leq 0
\end{align*}
where $J=P^{-1}S$ and $L=P^{-1}R$. Notice that when the diagonal elements of $Q$ computed by \eqref{eq:qb_compute} are very large, $Q^{-1}$ will have very small values. This will be compensated by having $T$ with very large eigenvalues in the second inequality, which will result in the violation of the first inequality.

In \cite{abbaszadeh2010}, \cite{zhang2012}, and \cite{phanomchoeng2011}, other parameterization-based designs, such as the one-sided Lipschitz property, quadratic inner boundedness, and bounded Jacobian, are presented in the context of Luenberger-like observers. However, the design criteria for one-sided Lipschitz and quadratic inner boundedness are found to be quite conservative; see \cite{niazi2022cdc} for more details about these limitations. Moreover, the design in \cite{phanomchoeng2011} that uses a bounded Jacobian property of $f$ is technically incorrect, as explained in \cite{niazi2022cdc}. Although \cite{rajamani2020} also uses a bounded Jacobian property for designing a Luenberger-like observer with switching observer gain, designing a switching signal remains an open problem. Finally, the design using the dissipativity properties of $\Tilde{f}(\eta,Hx)$ in \cite{moreno2004} results in nonlinear matrix inequality, which is difficult to solve computationally.

In addition to the conservativeness of the parameterization-based observer design criteria, finding parameterizations is also a difficult task for general nonlinear functions, especially for non-differentiable functions. Therefore, it is important to have observer design criteria that does not rely on any specific parameterization of the system's nonlinearity.

\section{Parameterization-free Observer Design}
\label{sec:par-free}

In this section, we introduce an observer design that does not rely on any parameterization of the nonlinearity $f$. This parameterization-free approach allows us to establish design criteria without requiring any global properties of $f$. Using this approach, we essentially treat the nonlinearity as an external disturbance in the error equation \eqref{eq:error}. As a result, the goal of the observer design is to stabilize the estimation error by effectively filtering out the nonlinearity by tracking it.

\begin{theorem}
    \label{thm:par-free}
    If there exist positive definite matrices $P,Q\in\bb{R}^{n_x\times n_x}$, a scalar $\rho\geq 0$, and gain matrices $J,L\in\bb{R}^{n_x\times n_y}$ and $K\in\bb{R}^{n_h\times n_y}$ such that
    \begin{subequations}
        \label{eq:general_asymptotic}
        \begin{align}
        \label{eq:general_asymptotic_a}
        \Phi + \Gamma^\T \Lambda \Gamma &\leq 0 \\
        \Lambda &\geq \rho I
        \label{eq:general_asymptotic_b}
        \end{align}
    \end{subequations}
    where
    \begin{align}
        \Phi &= \left[\ba{cc} 
        M^\T P + PM + Q & PNG \\ G^\T N^\T P & -\rho I_{n_f} \ea\right] \label{eq:Phi} \\
        \Gamma &= \left[\ba{cc}
        H-KC & 0_{n_h\times n_f} \\ 0_{n_f\times n_x} & I_{n_f} \ea\right] \label{eq:Gamma}
    \end{align}
    with $M,N$ given in \eqref{eq:matrices_M,N}, then the estimation error $\xi(t)$ globally asymptotically converges to zero.
\end{theorem}
\begin{proof}
    Note that $\eta(t)-Hx(t) = (H-KC)\xi(t)$ and, for every $\eta,Hx$ and every $\Lambda\geq 0$, it holds that
	\begin{equation}
        \left[\ba{c} \eta-Hx \\ \Tilde{f}(\eta,Hx) \ea\right]^\T \Lambda \left[\ba{c} \eta-Hx \\ \Tilde{f}(\eta,Hx) \ea\right]
        = \left[\ba{c} \xi \\ \Tilde{f}(\eta,Hx) \ea\right]^\T \Gamma^\T \Lambda \Gamma \left[\ba{c} \xi \\ \Tilde{f}(\eta,Hx) \ea\right] \geq 0
        \label{eq:GamLamGam}
    \end{equation}
    where $\xi(t)=\hat{x}(t)-x(t)$ is the estimation error, $\eta(t)$ is given in \eqref{eq:obs_eta}, $\Tilde{f}(\eta,Hx)$ in \eqref{eq:Tilde_f}, and $\Gamma$ in \eqref{eq:Gamma}.
    
    Let $V=\xi^\T P \xi$ be a candidate Lyapunov function. It satisfies Lemma~\ref{lem:lyap}(i) because
    \[
    \lambda_{\min}(P) \|\xi(t)\|^2 \leq \xi^\T(t) P \xi(t) \leq \lambda_{\max}(P) \|\xi(t)\|^2. 
    \]
    To show Lemma~\ref{lem:lyap}(ii), we take its time derivative and add and subtract certain terms on the right hand side
    \begin{align*}
        \Dot{V} &= (\xi^\T M + \Tilde{f}^\T G^\T N^\T) P \xi + \xi^\T P (M\xi + NG\Tilde{f}) \\
        &= \left[\ba{c} \xi \\ \Tilde{f} \ea\right]^\T \left[\ba{cc} 
        M^\T P + PM & PNG \\ G^\T N^\T P & 0_{n_f\times n_f} \ea\right] \left[\ba{c} \xi \\ \Tilde{f} \ea\right]^\T \pm \xi^\T Q \xi \pm \rho \Tilde{f}^\T \Tilde{f} \pm \left[\ba{c} \eta-Hx \\ \Tilde{f} \ea\right]^\T \Lambda \left[\ba{c} \eta-Hx \\ \Tilde{f} \ea\right].
    \end{align*}
    Using \eqref{eq:GamLamGam}, we obtain
    \begin{align}
        \Dot{V} &= -\xi^\T Q \xi + \left[\ba{c} \xi \\ \Tilde{f} \ea\right]^\T (\Phi + \Gamma^\T \Lambda \Gamma ) \left[\ba{c} \xi \\ \Tilde{f} \ea\right] + \left[\ba{c} \eta-Hx \\ \Tilde{f} \ea\right]^\T (-\Lambda+\rho\mc{I}) \left[\ba{c} \eta-Hx \\ \Tilde{f} \ea\right]
        \label{eq:Dot_V_general}
    \end{align}
    where $\Phi$ is given in \eqref{eq:Phi}, $\Gamma$ in \eqref{eq:Gamma}, and
    \[
    \mc{I} = \left[\ba{cc} 0 & 0 \\ 0 & I_{n_f} \ea\right].
    \]
    Notice that if \eqref{eq:general_asymptotic_b} holds, then $-\Lambda+\rho\mc{I}\leq 0$ because $\rho I \geq \rho \mc{I}$ for $\rho\geq 0$. Moreover, if \eqref{eq:general_asymptotic_a} holds as well, then $\Dot{V}\leq -\xi^\T Q \xi$. Since $Q>0$, we have $\Dot{V}<0$. Hence, by Lemma~\ref{lem:lyap}, the error equation \eqref{eq:error} is (uniformly) globally asymptotically stable, which concludes the proof.
\end{proof}

Note that by choosing some $\rho\geq 0$, substituting $R=PL$ and $S=PJ$, and using Schur complement lemma we can obtain an LMI feasibility problem equivalent to \eqref{eq:general_asymptotic} as
\begin{equation}
\label{eq:LMI_general}
    \boxed{P>0, \; Q>0, \; \left[\ba{cc} \Phi & \Gamma^\T \\ \Gamma & -\Tilde{\Lambda} \ea\right] \leq 0, \; \Tilde{\Lambda} \leq \frac{1}{\rho} I}
\end{equation}
where $\Tilde{\Lambda}=\Lambda^{-1}$.
In \eqref{eq:Phi}, $PM=PA-RCA-SC$ and $PN=P-RC$. Then, the gain matrices are obtained as $J=P^{-1}S$ and $L=P^{-1}R$.

For the global exponential stability of \eqref{eq:error}, we can have the following modifications to Theorem~\ref{thm:par-free}. 

\begin{corollary}
    If the condition \eqref{eq:general_asymptotic} is satisfied with $Q\geq \alpha P > 0$, for some $\alpha>0$, then the estimation error $\xi(t)$ globally exponentially converges to zero.
\end{corollary}
\begin{proof}
    Assume \eqref{eq:general_asymptotic} is satisfied, then \eqref{eq:Dot_V_general} implies $\Dot{V}\leq -\xi^\T Q \xi$. If $Q\geq \alpha P$, then $\Dot{V}\leq -\alpha V$. The result then follows from Lemma~\ref{lem:lyap}.
\end{proof}

For a given $\alpha>0$, the LMI condition \eqref{eq:LMI_general} can be obtained with $Q\geq \alpha P$ instead of $Q>0$. However, $\alpha>0$ can be made a decision variable in \eqref{eq:LMI_general}.

\begin{corollary}
    \label{corol:alpha}
    If the condition \eqref{eq:general_asymptotic} is satisfied with $Q\geq \alpha I_{n_x}$ and $P\leq I_{n_x}$, for some $\alpha>0$, then the estimation error $\xi(t)$ globally exponentially converges to zero.
\end{corollary}
\begin{proof}
    Assume \eqref{eq:general_asymptotic} is satisfied, then \eqref{eq:Dot_V_general} implies $\Dot{V}\leq -\xi^\T Q \xi$. If $Q\geq \alpha I_{n_x} \geq P$, then $\Dot{V}\leq -\alpha \xi^\T \xi \leq -\alpha V$. Then, the proof is concluded by Lemma~\ref{lem:lyap}.
\end{proof}

Using Corollary~\ref{corol:alpha}, the LMI condition \eqref{eq:LMI_general} is modified by adding $\alpha >0$, $Q\geq \alpha I$, and $0<P\leq I$.


In the next section, we demonstrate the effectiveness of parameterization-free approach for the state estimation of networked SIR epidemic processes.

\section{Application to a Networked SIR Epidemic Model}
\label{sec:simulation}

Consider a set of nodes $\mc{V}=\{1,\dots,n\}$ connected over a digraph $\mc{G}=(\mc{V},\mc{E})$, where $\mc{E}\subseteq\mc{V}\times\mc{V}$ is the set of edges. Every node~$i$ possesses a state $(x_{\S}^i(t),x_{\I}^i(t),x_{\R}^i(t))$ that represents the fraction of susceptible, infectious, and recovered individuals in its population. An edge $(i,j)\in\mc{E}$ implies that the susceptible population of node~$i$ can be infected by the infectious population of node~$j$, where the weight of such a connection is determined by the infection rate $\beta_{i}w_{ij}$, where $\beta_i>0$ is the infection susceptibility of $i$ and $w_{ij}>0$ is the contact rate of $i$ with $j$. The infectious population of each node~$i$ recovers with a recovery rate $\delta_i>0$.

Mathematically, the deterministic SIR epidemic process over $\mc{G}$ is described by
\begin{subequations}
    \label{eq:sir_net_i}
    \begin{align}
	   \Dot{x}_{\I}^i &= (1-x_{\I}^i-x_{\R}^i) \sum_{j=1}^n \beta_{i}w_{ij} x_{\I}^j - \delta_i x_{\I}^i \\
        \Dot{x}_{\R}^i &= \delta_i x_{\I}^i
    \end{align}
\end{subequations}
where $x_{\S}^i(t) = 1-x_{\I}^i(t)-x_{\R}^i(t)$.

Assume that $x_{\R}^i(t)$ is available through measurements at each node~$i$. Define $W=[w_{ij}]$ to be the weighted adjacency matrix of the digraph $\mc{G}$, and let $B=\diag(\beta_1,\dots,\beta_n)$ and $D=\diag(\delta_1,\dots,\delta_n)$. Let $x(t)=[\ba{cc} x_{\I}(t)^\T & x_{\R}(t)^\T \ea]^\T$ with $x_{\I}(t)=[\ba{ccc} x_{\I}^1(t) & \dots & x_{\I}^n(t) \ea]^\T$ and $x_{\R}(t)=[\ba{ccc} x_{\R}^1(t) & \dots & x_{\R}^n(t) \ea]^\T$, then \eqref{eq:sir_net_i} can be written as \eqref{eq:sys} with
\begin{align*}
    A &= \left[\ba{cc}
    BW-D & 0_{n\times n} \\ D & 0_{n\times n}
    \ea\right], \; G = \left[\ba{c} -I_n \\ 0_{n\times n}\ea\right], \; H=I_{2n} \\
    C &= \left[\ba{cc} 0_{n\times n} & I_n \ea\right]
\end{align*}
and $f(Hx)=\diag(x_\I(t)+x_\R(t))BWx_\I(t)$. 
The pair $(A,C)$ is observable if the recovery rates $\delta_i$ are distinct.

\begin{figure}[h]
    \centering
    \includegraphics[width=0.2\textwidth, trim={80 75 70 65}, clip]{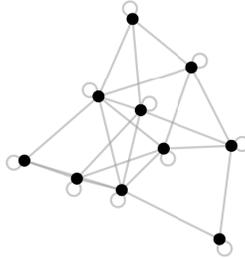}
    \caption{Bidirected graph $\mc{G}$ considered for the simulation example.}
    \label{fig:graph}
\end{figure}

We consider the number of nodes to be $n=10$. We generate a random bidirected graph $\mc{G}$ hown in Fig.~\ref{fig:graph} with the probability of edge between each pair of nodes equal to $0.5$, where the edge weight $w_{ij}$ is chosen in the interval $(0,2)$ uniformly randomly. Each node has a self-loop indicating infection spread within its population. The parameters $\beta_i$ and $\delta_i$ are also chosen uniformly randomly from $(0,1)$ for each $i=1,\dots,n$. For 100 realizations of random bidirected graphs, we found that the Lipschitz constant $\ell$ of nonlinearity $f(Hx)$ ranges from $7.5$ to $25.6$, which makes the Lipschitz-based design criterion \eqref{eq:ari_lip} infeasible given that $\ell^2$ would be very large. Therefore, in this case, the parameterization-free design criterion \eqref{eq:general_asymptotic} is preferable.

Choosing $\rho=1$, we use SeDuMi \cite{sturm1999} in MATLAB to solve the LMI \eqref{eq:LMI_general} and obtain the observer gain matrices $J,K,L$. Then, after using \eqref{eq:matrices_M,N} to obtain matrices $M$ and $N$, the observer \eqref{eq:obs} is initialized as $z(0)=-Ly(0)$. The measurement output $y(t)$ is assumed to be corrupted by a noise $v(t)\sim\mc{N}(0_{n_y},0.001 I_{n_y})$. The state estimation of each node is illustrated in Fig.~\ref{fig:estimation}, where it is important to point out that our observer is capable of handling the noise effectively in the estimation. The norm of the estimation error is plotted in Fig.~\ref{fig:error}.

\begin{figure}[t]
    \centering
    \includegraphics[width=0.35\textwidth, trim={0 0 0 0}, clip]{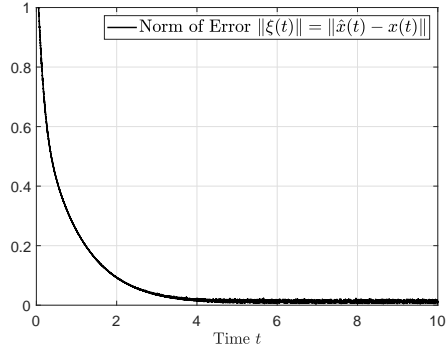}
    \caption{Norm of the estimation error.}
    \label{fig:error}
\end{figure}

\begin{figure}[t]
    \centering
    \includegraphics[width=\textwidth, trim={75 10 75 20}, clip]{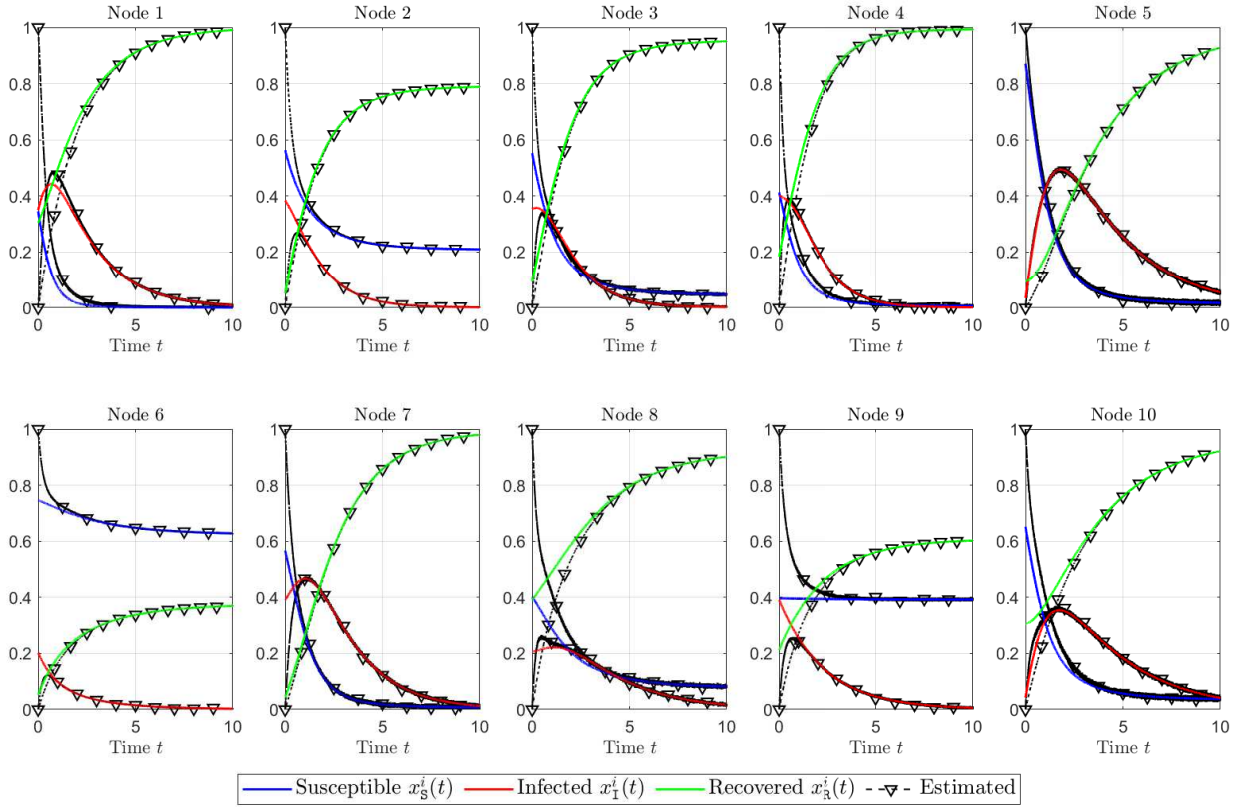}
    \caption{Epidemic state estimation in each node of the bidirected graph $\mc{G}$ depicted in Fig.~\ref{fig:graph}.}
    \label{fig:estimation}
\end{figure}

\section{Concluding Remarks}
\label{sec:conclusion}

We proposed a novel parameterization-free observer design approach for state estimation of nonlinear systems. We demonstrated that traditional methods relying on parameterized inequalities of the nonlinearity can lead to unnecessarily conservative observer design criteria when the parameterization is enforced globally. In contrast, our proposed approach makes no assumptions about the system's nonlinearity, rendering it parameterization-free.

We established observer design criteria that guarantee both asymptotic and exponential convergence of the estimation error to zero by treating the nonlinearity as an unknown disturbance in the estimation error dynamics, where the goal of the observer is to track not only the state but also the nonlinearity of the system. We also demonstrated the effectiveness of our approach by using it for the state estimation of a networked SIR epidemic model. This application is particularly important, as state estimation of epidemic models is crucial, and traditional parameterized methods become infeasible due to the quadratic nature of the nonlinearity.

Our proposed parameterization-free approach offers a promising alternative to traditional observer design methods and has the potential to be applied to a wide range of nonlinear systems. In conclusion, this work presents a significant contribution to the field of observer design for nonlinear systems and opens up new possibilities for state estimation in various applications. The future work includes a systematic method to handle process and measurement noises in the nonlinear system.

\bibliography{references}
\bibliographystyle{IEEEtran}

\end{document}